\documentclass[
    prl,
    aps,
    twocolumn,
    superscriptaddress,
    notitlepage,
    10pt,
final,   
]{revtex4-1}

\usepackage{graphicx}
\usepackage{dcolumn}
\usepackage{bm}
\usepackage{amssymb}
\usepackage{amsmath}
\usepackage{physics}
\usepackage{sublabel}
\usepackage[utf8]{inputenc}
\usepackage{hyperref}
\usepackage[ruled,vlined]{algorithm2e}
\usepackage[usenames, dvipsnames]{color}
\usepackage[english]{babel}
\usepackage[T1]{fontenc}
\usepackage{bbm}
\usepackage{amsmath,amssymb,amsthm}
\theoremstyle{definition}

\newtheorem{theorem}{Theorem}[section]

\newtheorem{claim}{Claim}[theorem]

\theoremstyle{remark}

\newcommand{\argmax}{\mathop{\mathrm{argmax}}}
\usepackage{float}
\usepackage{verbatim}
\hypersetup{colorlinks = true, linkcolor=blue, citecolor=blue, urlcolor=blue}

  \makeatletter
    \renewcommand\@make@capt@title[2]{%
     \@ifx@empty\float@link{\@firstofone}{\expandafter\href\expandafter{\float@link}}%
      {\textsc{#1}}\@caption@fignum@sep#2\quad}%
    \makeatother

\begin{document}


\title{Shallow-Depth Variational Quantum Hypothesis Testing}

\author{Mahadevan Subramanian}
\affiliation{Department of Physics, Indian Institute of Technology Bombay, Powai, Mumbai 400076, India}
\author{Sai Vinjanampathy}
\email[]{sai@phy.iitb.ac.in}
\affiliation{Department of Physics, Indian Institute of Technology Bombay, Powai, Mumbai 400076, India}
\affiliation{Centre of Excellence in Quantum Information, Computation, Science and Technology, Indian Institute of Technology Bombay, Powai, Mumbai 400076, India.}
\affiliation{Centre for Quantum Technologies, National University of Singapore, 3 Science Drive 2, Singapore 117543, Singapore.}


\date{\today}
\begin{abstract}
 We present a variational quantum algorithm for differentiating several hypotheses encoded as quantum channels. Both state preparation and measurement are simultaneously optimized using success probability of single-shot discrimination as an objective function which can be calculated using localized measurements. Under constrained signal mode photon number quantum illumination we match the performance of known optimal 2-mode probes by simulating a bosonic circuit. Our results show that variational algorithms can prepare optimal states for binary hypothesis testing with resource constraints. Going beyond the binary hypothesis testing scenario, we also demonstrate that our variational algorithm can learn and discriminate between multiple hypotheses.
\end{abstract}

\maketitle


\paragraph{Introduction.---}
Quantum resources have been implied in the enhanced performance of technology tasks in the future such as computing, communication and hypothesis testing. In view of this, the presence of noisy intermediate scale quantum  (NISQ) computers presents an opportunity to go beyond algorithms that are classical in nature \cite{preskill2018quantum,bharti2022noisy}. Variational quantum algorithms (VQAs) \cite{cerezo2021variational,mcclean2016theory} are algorithms that solve optimization problems by evaluating a cost function using a parameterized quantum circuit (PQC) and updating parameters using classical optimization. Though VQAs have a wide array of applications \cite{Kandala_2017,peruzzo2014variational,PhysRevA.99.062304,PhysRevLett.127.220504,lau2022convex,PhysRevX.11.041045,marciniak2022optimal,PhysRevResearch.4.013083,chen2021variational}, their limitations such as barren plateaus in the optimization landscape \cite{mcclean2018barren,wang2021noise} have sparked discussion on reducing the depth of these circuits, and the use of local observables to avoid them \cite{Cerezo_2021}. Hence, shallow-depth circuits (circuits with depth $\mathcal{O}(1)$ or $\mathcal{O}(\log n)$ for $n$ qubits \cite{Cerezo_2021}) with a demonstrable advantage are highly sought after.\\
\indent In this manuscript, we propose a novel application of VQAs by demonstrating that shallow depth circuits can be employed for performing quantum channel discrimination. We apply our algorithm to the problem of quantum illumination \cite{doi:10.1126/science.1160627,2009Shapiro,2008Tan,PhysRevLett.118.070803,PhysRevA.98.012319,PhysRevA.103.012411,PhysRevResearch.2.023414}, where we show that low-depth parameterized quantum circuits (PQCs) can obtain known optimal values for Chernoff bound and trace distance in the case of quantum illumination. We begin by establishing results for discrimination between quantum channels on $n$ qubits in the following section and extend our analysis numerically by simulating bosonic modes for quantum illumination. 
\paragraph{Variational Quantum Hypothesis Testing.---} Our key insight is that shallow depth VQAs can perform binary hypothesis testing of generic quantum channels. We assert that PQCs are able to recover known optima by providing analytical results for arbitrary quantum hypothesis testing (QHT) tasks and numerical simulations for the task of Gaussian quantum illumination. We briefly review channel discrimination, which aims at distinguishing two generic quantum channels $\mathcal{E}_0:L(\mathcal{H})\to L(\mathcal{H})$ and likewise $\mathcal{E}_1$ using measurements, with the set of linear operators defined as $L(\mathcal{H})$. The null hypothesis $H_0$ corresponds to $\mathcal{E}_0$ and the alternate hypothesis $H_1$ to $\mathcal{E}_1$. An input quantum state $\rho\in D(\mathcal{H}\otimes\mathcal{H})$ is sent to one of two channels creating the output states $\rho_{i} = (\mathcal{E}_{i}\otimes \mathbb{I})(\rho)$ ($i = 0$ or $1$) \cite{aharonov1998quantum}. Following this, a measurement is performed over these output states using the positive operator-valued measure (POVM) $\{\Gamma,\mathbb{I}-\Gamma\}$ where $\Gamma$ is a valid POVM element in $L(\mathcal{H}\otimes\mathcal{H})$. The measurement outcome corresponding to $\Gamma$ is the acceptance criteria for hypothesis $H_0$ and outcome $\mathbb{I}-\Gamma$ is the acceptance criteria for $H_1$. The type-I error (false positive) is defined as $\alpha = 1-\text{Tr}(\Gamma\rho_0)$ and type-II error (false negative) is defined as $\beta = \text{Tr}(\Gamma\rho_1)$. Moving forward, we optimize the total error probability $(\alpha + \beta)/2$ under the assumption that either channel has equal likelihood to be applied.\\ 
\indent In our protocol, one of the two channels is acted on the prepared state $\rho$ following which the outcome of the POVM $\{\Gamma,\mathbb{I}-\Gamma\}$ is used to determine which channel had been applied. Optimal strategies differ based on whether the channel can be applied sequentially or in parallel \cite{PhysRevA.81.032339,PhysRevLett.127.200504}. In a general parallel strategy \cite{PhysRevLett.127.200504}, one would be discriminating between the maps $\mathcal{E}^{\otimes n}_0$ and $\mathcal{E}^{\otimes n}_1$. We examine a limited parallel strategy where the channel is applied on $n$ copies of the initial state to better study asymptotic behavior. The Holevo-Helstrom result bounds the total error probability for $\rho_0^{\otimes n}$ and $\rho_1^{\otimes n}$. This asymptotically decays exponentially with the exponent determined by the Chernoff bound between $\rho_0$ and $\rho_1$ \cite{PhysRevLett.98.160501,10.1214/08-AOS593}. In the case of symmetric hypothesis testing \cite{2014Li,hiai1991proper,887855}, the minimum error probability for a single copy use of the map is $1/2 - \|\mathcal{E}_0-\mathcal{E}_1\|_\diamond/4$ where $\|\mathcal{E}_0-\mathcal{E}_1\|_\diamond$ is the diamond distance between the two channels. This is defined as
\begin{equation}\label{eq:diamond}
    \|\mathcal{E}_0 - \mathcal{E}_1\|_\diamond = \sup_{\rho}(\|(\mathcal{E}_0\otimes\mathbb{I}_{\mathcal{H}})(\rho) - (\mathcal{E}_1\otimes\mathbb{I}_{\mathcal{H}})(\rho)\|_1),
\end{equation}
where $\|A\|_1 = \text{Tr}\sqrt{A^\dagger A}$ is the trace norm. This bound is a direct extension of the Holevo-Helstrom bound \cite{holevo1973bounds,Helstrom1969}. The diamond distance provides a fundamental bound which holds true for all states in $D(\mathcal{H}\otimes\mathcal{H}')$ with $\dim(\mathcal{H}')\geq\dim(\mathcal{H})$ \cite{aharonov1998quantum} and hence optimizing over states in $D(\mathcal{H}\otimes\mathcal{H})$ is sufficient to reach an optimal probe. 
\begin{figure}[ht]
    \centering    \includegraphics[width=\linewidth]{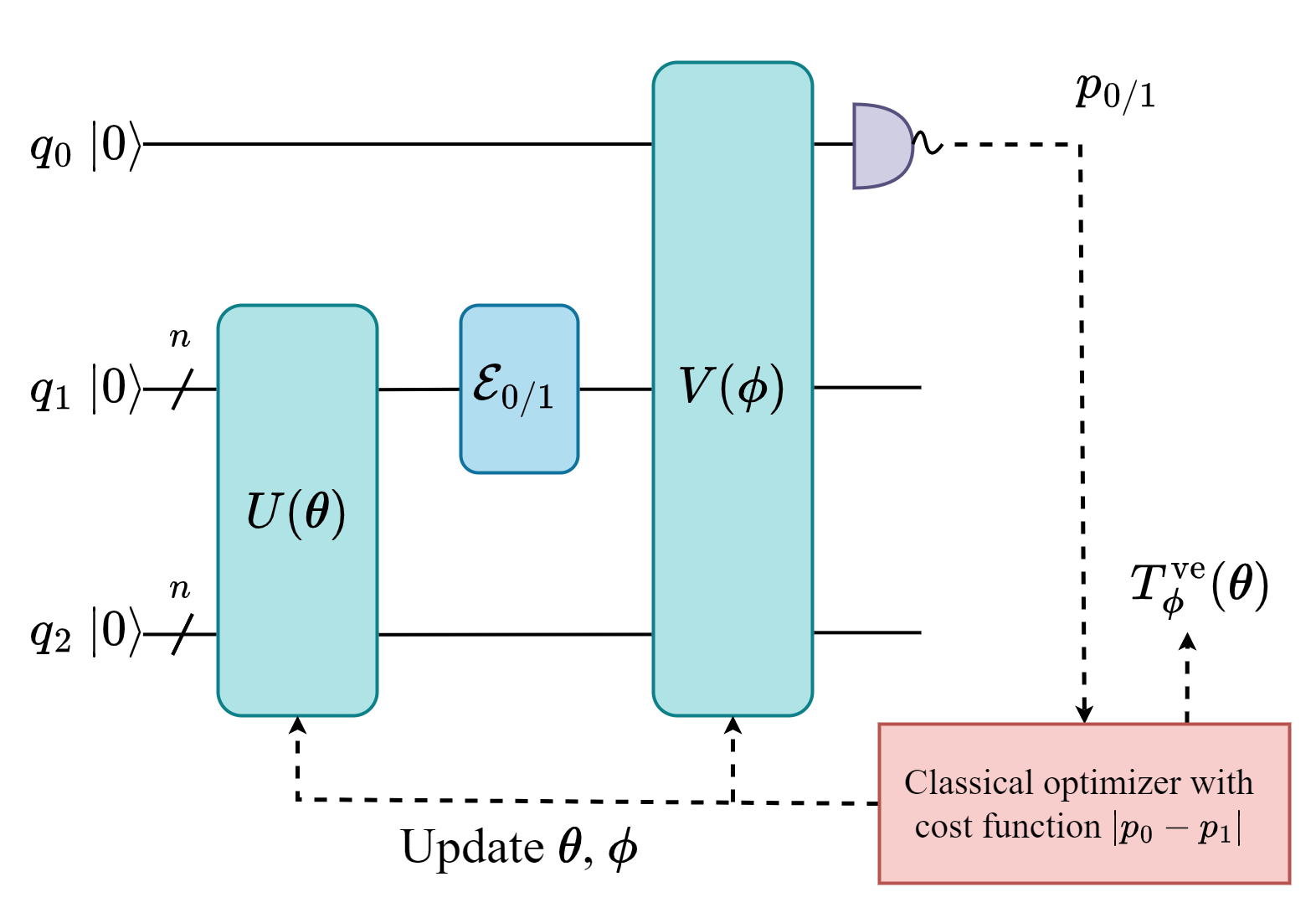}
    \caption{Circuit for variational QHT. The state preparation is over the two $n$ qubit registers (or modes) $q_1$ and $q_2$ and the map is applied on one mode. An ancillary qubit (or mode) $q_0$ is used for measurements.}
    \label{fig:circuit1}
\end{figure}

\indent For our variational algorithm, we prepare an entangled probe state $\rho_{\pmb{\theta}}$ in registers $q_1$ and $q_2$ using the parameterized unitary $U(\pmb{\theta})$ as depicted in Fig. \ref{fig:circuit1}. This is motivated by the advantage of entangled probe states \cite{PhysRevLett.87.270404,PhysRevLett.122.140404,PhysRevLett.102.250501,aharonov1998quantum} even in the case of entanglement breaking maps \cite{PhysRevA.72.014305}. Either the map $\mathcal{E}_0$ or $\mathcal{E}_1$ is then applied on the register $q_1$. The state after applying the map is $\rho_{i,\pmb{\theta}} = (\mathcal{E}_{i}\otimes \mathbb{I})(\rho_{\pmb{\theta}})$ ($i=\{0,1\}$).\\
\indent The optimal protocol is achieved by maximizing the trace distance $\|\rho_{0,\pmb{\theta}}-\rho_{1,\pmb{\theta}}\|_1$, as evidenced from Eq~\eqref{eq:diamond} and is known to be difficult \cite{chen2021variational,d2003quantum,watrous2008quantum,https://doi.org/10.48550/arxiv.2108.08406} to evaluate. We get around this by using an estimate of the trace distance as an objective function which makes use of a two-outcome POVM $\{\Gamma,\mathbb{I}-\Gamma\}$. The ancillary qubit $q_0$ (see Fig. \ref{fig:circuit1}) is used for measurements using the parameterized unitary $V(\pmb{\phi})$ that encodes a Naimark extension \cite{wilde2013sequential} of the POVM $\{\Gamma,\mathbb{I}-\Gamma\}$. The probability for outcome $\Gamma$ is $p_{i} = \text{Tr}(\Gamma\rho_{i,\pmb{\theta}}) = \text{Tr}\left((\ketbra{0}{0}\otimes I)V(\pmb{\phi})(\ketbra{0}{0}\otimes \rho_{i,\pmb{\theta}})V(\pmb{\phi})^\dagger\right)$. Outcome $\Gamma$ is the acceptance criteria for hypothesis $H_0$ and $\mathbb{I}-\Gamma$ is the acceptance criteria for hypothesis $H_1$. Hence, $(p_0 + 1 - p_1)/2$ corresponds to the success probability in the classification task. We define a variational estimate of the trace distance $T^{\mathrm{ve}}_{\pmb{\phi}}(\pmb{\theta})$, which is always bounded above by the true trace distance and this bound is saturated when $V(\pmb{\phi})$ encodes a Naimark extension of the Helstrom POVM \cite{chen2021variational}. The variational estimate of the trace distance is defined as 
\begin{equation}\label{eq:TDest}
    T^{\mathrm{ve}}_{\pmb{\phi}}(\pmb{\theta}) = |p_0-p_1| \leq \frac{1}{2}\|\rho_{0,\pmb{\theta}} - \rho_{1,\pmb{\theta}}\|_1.
\end{equation}
\indent Our algorithm updates parameters to maximize $T^{\mathrm{ve}}_{\pmb{\phi}}(\pmb{\theta})$. $T^{\mathrm{ve}}_{\pmb{\phi}}(\pmb{\theta})$ is evaluated with a certain value of $\pmb{\theta}$ and $\pmb{\phi}$ using the quantum computer. These values are then fed into a classical optimizer which proposes updated parameter values for the next iteration and this is repeated till the convergence condition for $T^{\mathrm{ve}}_{\pmb{\phi}}(\pmb{\theta})$ is met, after which the values of $\pmb{\theta}$ and $\pmb{\phi}$ are reported. The optimization aims to approach $\pmb{\theta}_0,\pmb{\phi}_0 = \argmax_{\pmb{\theta},\pmb{\phi}}(T^{\mathrm{ve}}_{\pmb{\phi}}(\pmb{\theta}))$. Maximizing the value of $T^{\mathrm{ve}}_{\pmb{\phi}}(\pmb{\theta})$ clearly maximizes the success probability of discrimination.

\indent While $T^{\mathrm{ve}}_{\pmb{\theta}}(\pmb{\phi})$ simplifies the optimization of trace distance by replacing true trace distance with $T^{\mathrm{ve}}$, the numerous parameters in the definitions of $U(\pmb{\theta})$ and $V(\pmb{\phi})$ may result in sub-optimal results. Suboptimality can arise in two ways, the first of which is that the optimization may fail to find the optimal states even if $U(\pmb{\theta})$ is expressible enough for the optimal states. The second source of suboptimality could be that though the qubit measurements produce optimal states according to $T^{\mathrm{ve}}_{\pmb{\theta}}(\pmb{\phi})$, these optimized probe states might have a small trace distance evaluated by $T(\rho_{0,\pmb{\theta}},\rho_{1,\pmb{\theta}})$. The second source of suboptimality is different from the first since an expressible $U(\pmb{\theta})$ is unrelated to the expressibility of $V(\pmb{\phi})$, which is simultaneously required to ensure that the ``good states are recognized''.\\
\indent We now proceed to prove that if $U(\pmb{\theta})$ and $V(\pmb{\phi})$ are sufficiently expressible, the optimal states generated by the VQAs indeed optimize the real trace distance. This is to verify that our shallow depth circuits are sufficiently expressible for both state preparation and the measurements needed for hypothesis testing. To see this, consider a fixed $U(\pmb{\theta})$ which generates a (perhaps) sub-optimal state, and then consider the optimization of $V(\pmb{\phi})$. Since $V(\pmb{\phi})$ represents the Naimark extension for an ideal two-outcome POVM, if $V(\pmb{\phi})$ is expressible enough, then it is guaranteed that the optimization of $T^{\mathrm{ve}}_{\pmb{\theta}}(\pmb{\phi})$ converges to the real trace distance $T(\rho_{0,\pmb{\theta}},\rho_{1,\pmb{\theta}})$ (further discussed in the supplemental information \cite{supplemental}, which includes references \cite{ben2009complexity,benenti2010computing,watrous2007distinguishing,JOHANSSON20121760,JOHANSSON20131234,ferraro2005gaussian,paris2012modern}). If $U(\pmb{\theta})$ is expressible, then the globally optimal states are expressible by the VQA. Judging the requirements on the expressibility of $U(\pmb{\theta})$ can be done by considering the state preparation to occur using a control pulse $\gamma(t)$ which can be represented as $b_\gamma$ classical bits. If the ideal probe state is reachable in polynomial time and we prepare $\rho_{\pmb{\theta'}}$ such that $\|\rho_{\pmb{\theta}'} - \rho_{\pmb{\theta}_0}\|\leq \epsilon$, the true trace distance $T(\rho_{0,{\pmb{\theta}'}},\rho_{1,{\pmb{\theta}'}}) \in [d(1-\epsilon),d]$ where $d = T(\rho_{0,{\pmb{\theta}_0}},\rho_{1,{\pmb{\theta}_0}})$ is the maximal trace distance. We know \cite{PhysRevLett.113.010502} that $b_\gamma$ scales with $\log(1/\epsilon)$ and the dimension of the manifold of polynomial time reachable states and this gives us an estimate of how much information is needed for having a good state preparation (see \cite{supplemental} for further details). This completes the analysis that our VQAs can indeed theoretically find the optimal states and measurements to perform the hypothesis test. We now apply our ideas to the example of Gaussian variational quantum illumination to show a practical application of this theoretical result.

\paragraph{Gaussian Variational Quantum Illumination.---} 
Quantum illumination \cite{doi:10.1126/science.1160627,2009Shapiro,2008Tan,PhysRevLett.118.070803,PhysRevA.98.012319,PhysRevA.103.012411,PhysRevResearch.2.023414} is the task of using entangled light to find out if a weakly reflective beam-splitter is present in a bright thermal bath. The map acting on a two-mode bosonic state $\rho_{SI}$ is given as $\mathcal{E}_\eta(\rho_{SI}) = \text{Tr}_S(U_\eta(\rho_B\otimes\rho_{SI})U_\eta^\dagger)$ where $U_\eta = \exp(i\sin^{-1}(\eta)(a_S^\dagger a_B - a_S a_B^\dagger))\otimes\mathbb{I}_I$ and $\rho_B$ is the thermal state with an average of $N_B$ photons. The hypothesis $H_0$ corresponds to the object being absent or equivalently beam-splitter has zero reflectivity hence is the map $\mathcal{E}_{\eta=0}$, and the hypothesis $H_1$ corresponds to the object being present with some weak reflectivity $r$ hence is the map $\mathcal{E}_{\eta=r}$. Despite clearly being an entanglement breaking map, there is an advantage in using a signal-idler entangled state as first demonstrated in \cite{doi:10.1126/science.1160627}. This advantage holds true even in the limit of $N_B\gg N_S$ where $N_S$ is the average number of signal photons \cite{2009Shapiro, 2008Tan}.\\
\begin{figure}[ht]
    \centering
    \includegraphics[width=\linewidth]{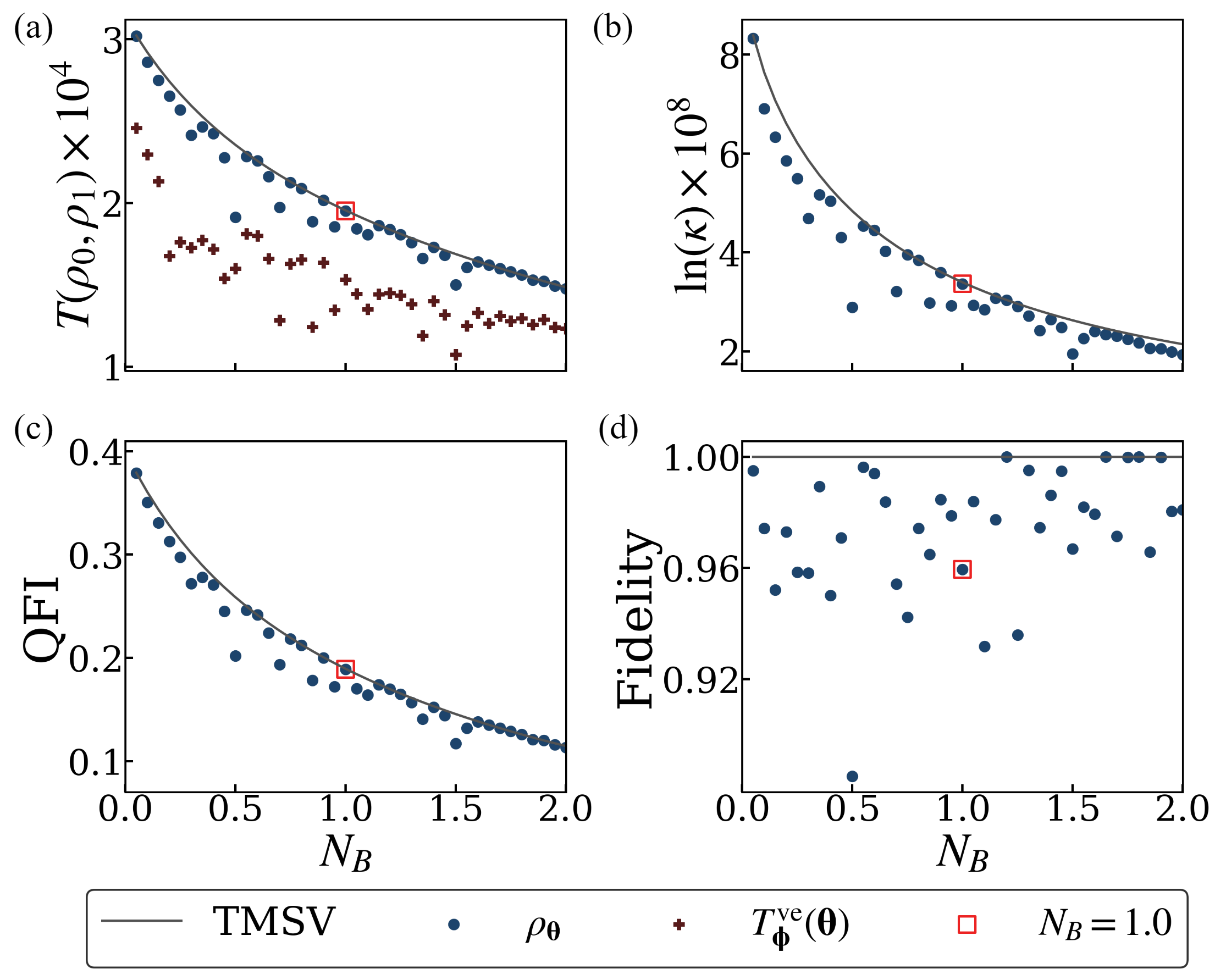}
    \caption{Simulation results with constrains mentioned in the main text. Subfigure (a) compares trace distance (blue dots are true trace distance and brown symbols are estimated trace distance), (b) compares Chernoff bound of the optimized state to TMSV, (c) compares QFI and (d) shows fidelity of the optimized state to TMSV. The red square represents shows an inequivalent optimized state with performance equivalent to the TMSV.}
    \label{fig:resultsVac}
\end{figure}
The probe's performance is judged by constraining the average signal photon number to be $N_S$. This relates the error probability to the energy constrained diamond norm \cite{winter2017energy}. All the results in this section are in the regime of a signal mode having a constrained photon number of $N_S$. Under this constraint, the optimal probe for the task of quantum illumination for single-shot discrimination is proven to be the two-mode squeezed vacuum state (TMSV) as shown in \cite{PhysRevA.98.012101,PhysRevA.103.062413}. The optimality of this state comes from the fact that it saturates the Chernoff bound \cite{PhysRevA.103.062413} as well as the quantum relative entropy \cite{PhysRevA.98.012101} making it suitable for both symmetric and asymmetric hypothesis testing. Even when compared to non-Gaussian states such as a photon-added TMSV \cite{https://doi.org/10.48550/arxiv.2110.06891}, TMSV remains optimal for fixed $N_S$. We note that TMSV state is suboptimal for the Helstrom bound \cite{PhysRevA.103.062413}.\\
\indent Our circuit for the task of variational quantum illumination follows the same protocol as that described in Fig. \ref{fig:circuit1} with $q_{0/1}$ being identified as bosonic modes signal $S$ and idler $I$ respectively. We parameterize the unitary $U$ as coherent displacements followed by two mode squeezing each by variable parameters and the unitary $V$ as coherent displacements followed by conditional phase gates followed by beam-splitters. When optimised over these resources, we note that our results compare favourably with known optimal results as seen in Fig. \ref{fig:resultsVac}. All unitary Gaussian transformations are representable in the Bloch-Messiah decomposition \cite{PhysRevA.94.062109}. While this is sufficient for complete parameterization for $n$-mode Gaussian states, we opt for a hardware-efficient ansatz \cite{Kandala_2017} to see the performance of the algorithm in the restricted setting of having $U$ and $V$ both be composed of single mode operations such as displacements followed by limited two-mode operations such as beam-splitters.\\
\indent The unitary $V$ encodes the Naimark extension of a two outcome POVM by having the vacuum state measurement on mode $q_0$ correspond to applying a joint POVM on the signal-idler state. One can always construct a unitary transformation acting on $q_0\otimes q_1\otimes q_2$ to transform a projective measurement on $q_0$ to a projective measurement in the space of $q_1\otimes q_2$ \cite{Buscemi_2005}. The procedure to construct unitary transformations using purely Gaussian operations is summarized in the supplemental information \cite{supplemental}.\\
\begin{figure}[ht]
    \centering
    \includegraphics[width=0.9\linewidth]{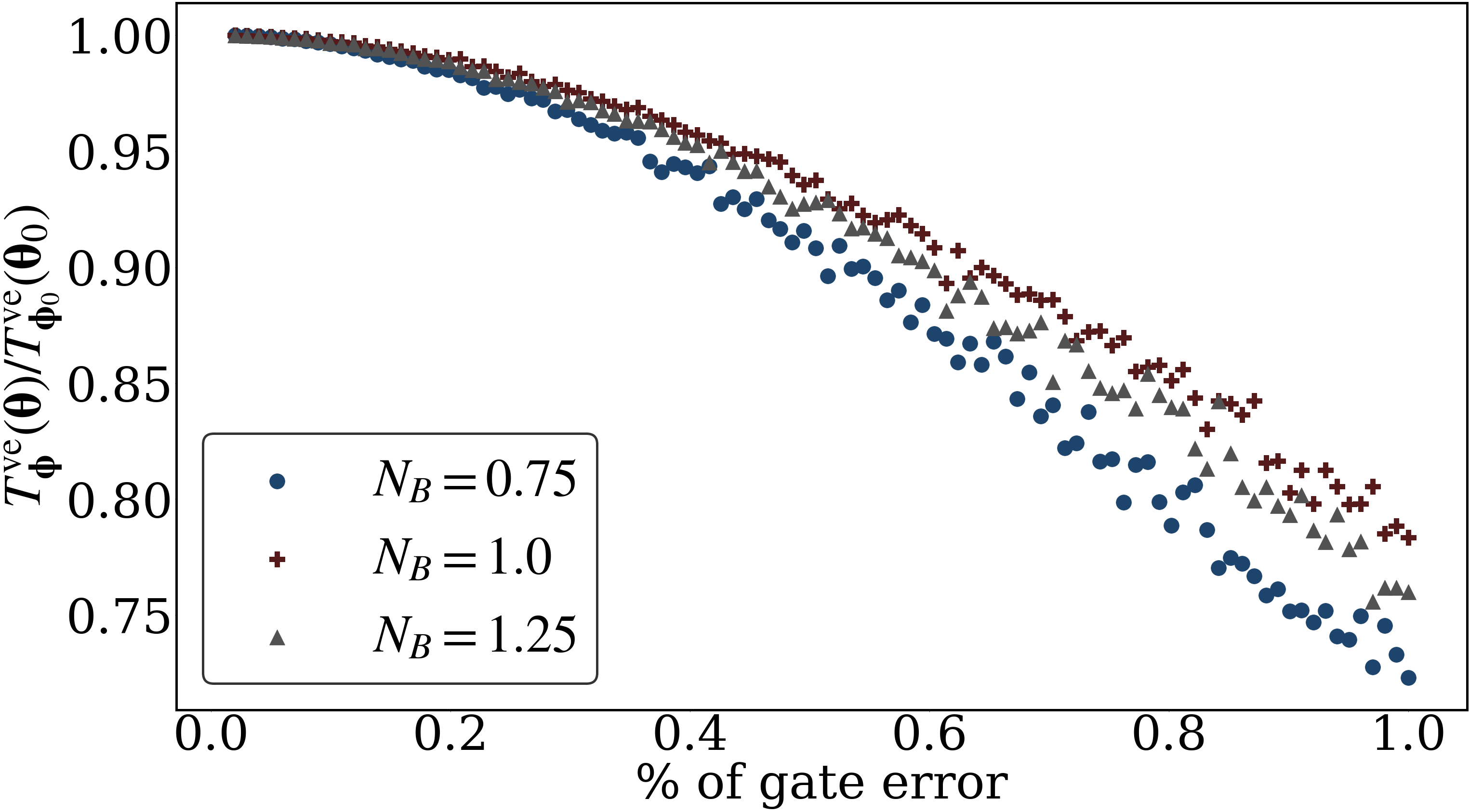}
    \caption{Noise robustness of variational quantum illumination showing Gaussian noise variance in gate error against the ratio of $T^{\mathrm{ve}}_{\pmb{\phi}}(\pmb{\theta})$ to the value it takes with zero gate error, optimized with the constraint that $N_S = 0.1$.}
    \label{fig:error}
\end{figure}
Our simulation (see \cite{supplemental} for code) over bosonic modes \cite{bromley2020applications,killoran2019strawberry} maximizes the function $T^{\mathrm{ve}}_{\pmb{\theta}}(\pmb{\phi})$. In Fig.~\ref{fig:resultsVac}(a) we highlight both the actual trace distance (blue dots) and the estimated trace distance (brown symbols) after optimization. The optimization regularly is able to find the true optimum, in this case given by the TMSV state. Fig.~\ref{fig:resultsVac}(b) shows that the optimality of TMSV for Chernoff bound is nearly matched with our low depth circuit that optimizes $T^{\mathrm{ve}}_{\pmb{\theta}}(\pmb{\phi})$. This shows that despite optimizing for single-copy discrimination, it is able to optimize to the state which is asymptotically optimal for $n$ copies. Given that the family of maps $U_\eta$ is continuous, our technique can be considered as a quantum sensing task \cite{PhysRevLett.108.170502}. In the limit of $N_S\to0$, the optimal quantum Fisher information (QFI) for this sensing task is given by the TMSV \cite{PhysRevLett.118.070803}. In Fig. \ref{fig:resultsVac}(c) we compare the QFI of the optimized state to this optimal value and observe that it is a good quantum sensor input despite being trained for only a fixed value of $\eta$. Fig. \ref{fig:resultsVac}(d) shows that despite reaching near the performance of the TMSV, the optimized states are not exactly equal to the TMSV, implying the existence of a manifold of states that perform just as well as the TMSV. We consider an example with $N_B = 1.0$ (red square) which matches the performance of the TMSV constrained with $N_S=N_I=0.1$. This state has average photons in idler and signal as $N_I \approx 0.17$ and $N_S = 0.1$  as well as differences in the reduced entropy and the coefficients of the Schmidt decomposition from the TMSV (see the supplemental information \cite{supplemental} for details), demonstrating that it isn't equivalent upto local one mode unitary transformations. To test noise resilience, we plot the averaged objective while the optimized circuit parameters have Gaussian noise (see \ref{fig:error}) and find favorable scaling with the error percentage. Our protocol can be trained for in one parameter regime and deployed in another, as discussed in the supplemental information \cite{supplemental}. Additionally, we show a way to construct a near-optimal Gaussian POVM for quantum illumination which supplements existing detection schemes \cite{9087936}.
\paragraph{Multiple Hypothesis Testing.---} We can extend the use of our algorithm to the case of multiple channel discrimination \cite{Vazquez_Vilar_2016, 1055351}. We have quantum channels $\mathcal{E}_1,\mathcal{E}_2,\dots \mathcal{E}_k$. Extending our earlier description, we now discriminate $\rho_i = (\mathcal{E}_i\otimes\mathbb{I})$ using the POVM set $\{\Pi_1,\dots\Pi_k\}$. Maximizing the success probability $P_{\mathrm{success}} = \sum_{i=1}^{k}\text{Tr}(\Pi_i\rho_i)/k$ (assuming each channel to be equally likely) is a semidefinite program over positive semidefinite variables $\rho$ and $\{\Pi_i\}$ with constraints $\text{Tr}(\rho) = 1$ and $\sum_{i}\Pi_i = \mathbb{I}$. Defining the ancilla $q_0$ as a qudit with $k$ levels makes it possible to encode a Naimark extension of this POVM set \cite{wilde2013sequential}. Proceeding with the same hybrid algorithm described for the binary case, we can now optimize over $\pmb{\theta}$ and $\pmb{\phi}$ to obtain the optimal state and measurement for this task.\\
\begin{figure}[ht]
    \centering
    \includegraphics[width=0.9\linewidth]{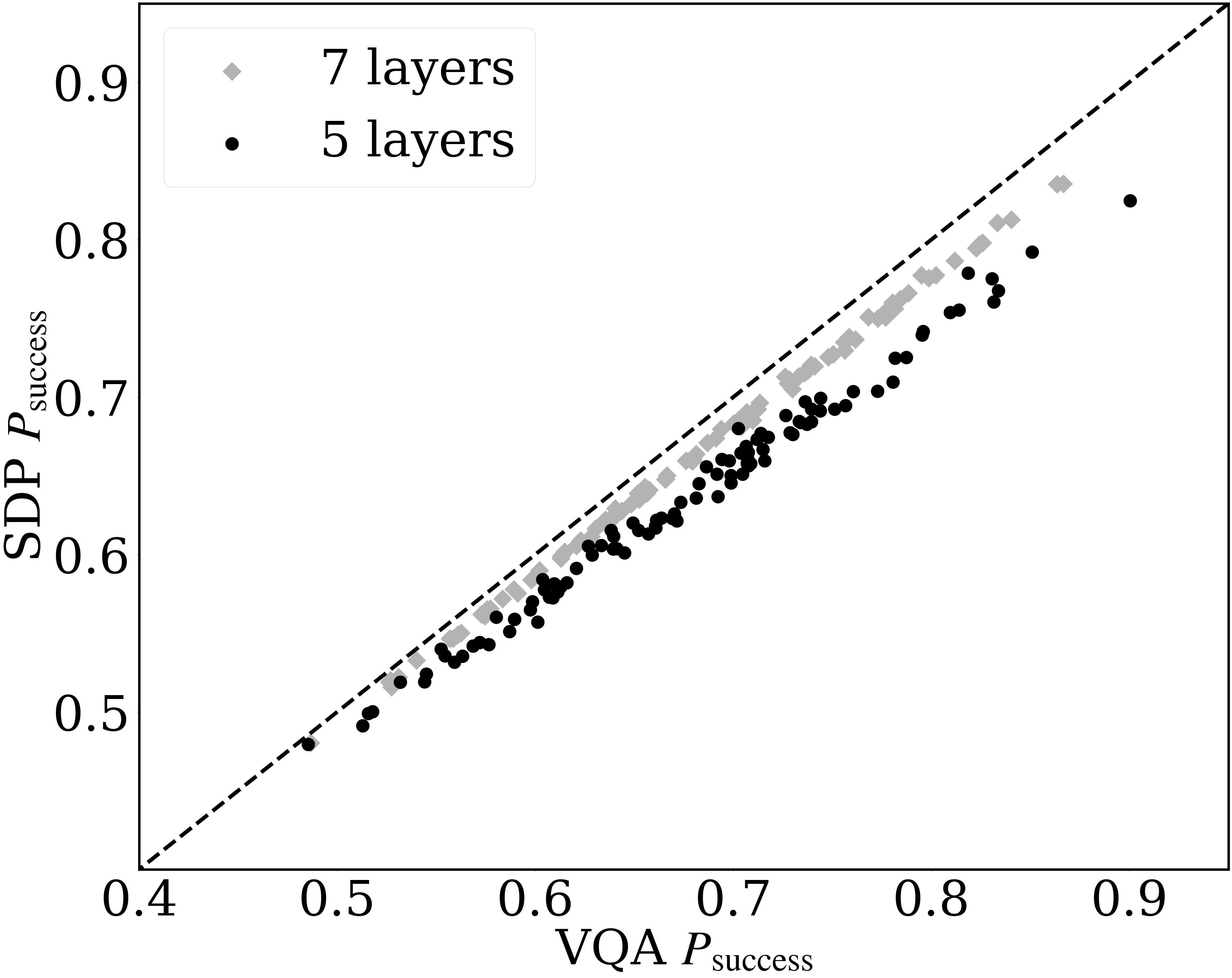}
    \caption{Success probability of VQAs for ternary hypothesis test of two-qubit channels over four qubits as described above with a qutrit ancilla. Ansatz $V(\pmb{\phi})$ (of 5 or 7 layers) is optimized while $U(\pmb{\theta})$ is fixed and prepares a GHZ state. The VQA outcome is compared against classical SDP result, which can be performed only for small system sizes.}
    \label{fig:ternary}
\end{figure}
\indent We explore a basic case of telling apart three different 2-qubit unitary channels. There are known analytical results for the multiple state discrimination of asymptotically many copies \cite{Li_2016,PhysRevLett.129.180502}. However there aren't known results for the case of single-shot discrimination. Hence we compare the performance of our VQA with a convex optimization done using the cvxpy package \cite{diamond2016cvxpy,agrawal2018rewriting}. Our results are shown in Fig. \ref{fig:ternary} showing that our algorithm is capable of near optimal performance even in the case of multiple hypotheses.
\paragraph{Discussion.---} Hypothesis testing is a task central to probability theory for the discrimination of probability distributions \cite{scott2005neyman}. QHT extends this to quantum channel discrimination and quantum state discrimination, and our work shows how VQAs can be applied for this task. The use of quantum resources have demonstrated an advantage here and have applications in exoplanet spectroscopy \cite{PhysRevA.107.022409}, superresolution between two incoherent optical sources \cite{PhysRevX.6.031033}, and as discussed in this work, quantum illumination \cite{doi:10.1126/science.1160627,2009Shapiro,2008Tan,PhysRevLett.118.070803,PhysRevA.98.012319,PhysRevA.103.012411,PhysRevResearch.2.023414}. Our VQA when applied for quantum illumination, is able to find a probe state that matches the optimal performance under signal photon number constraint (see Fig. \ref{fig:resultsVac}). These results demonstrate noise resilience as shown in Fig. \ref{fig:error} supporting the experimental viability of such algorithms for near-term quantum devices \cite{preskill2018quantum,bharti2022noisy}.\\
\indent Our results on QHT have direct application in general channel discrimination \cite{PhysRevLett.101.180501,PhysRevLett.101.180501} and tasks such as quantum reading \cite{PhysRevLett.106.090504,doi:10.1126/sciadv.abc7796}. It must be noted that the related generalized tasks such as the testing of a quantum channel to be an isometry, the determination of distinguishability between $k$ input states for a quantum channel, or checking equality of two unitaries are QMA-complete \cite{rosgen2010testing, beigi2007complexity,janzing2005non}. We find that the information content of a control pulse required to make the trace distance between the states lie in $[d(1-\epsilon),d]$ (where $d$ is the maximum possible trace distance between the states obtained on sending a probe state through the channels) scales as $\log(1/\epsilon)$ if the optimal state is polynomially reachable. Hence our algorithm would also take exponential resources to converge for difficult tasks, which situates general complexity-theoretic results within our framework.\\
\indent Our algorithm nonetheless performs well as benchmarked by channel discrimination measures. The algorithm achieves near optimal results for the discrimination between two-qubit unitaries \cite{supplemental}. Our work can also be clearly generalized to sequential and parallel channel discrimination \cite{PhysRevLett.127.200504}. These have practical applications in a variety of tasks such as quantum metrology and certifying quantum circuits \cite{cert_cqt}. We speculate that via the Pinsker's inequality \cite{https://doi.org/10.48550/arxiv.2005.04553} which lower bounds quantum relative entropy in terms of the trace distance, we can generalize our algorithm to asymmetric QHT \cite{PhysRevA.90.052307}. Hence variational QHT is sure to find disparate applications in future quantum technologies.

\begin{acknowledgments}
\paragraph*{Acknowledgments.---} S.V. acknowledges support from Government of India DST-QUEST grant number DST/ICPS/QuST/Theme-4/2019 and thanks Saikat Guha for valuable discussion on discriminating multiple hypotheses.
\end{acknowledgments}

\clearpage
\appendix
\onecolumngrid

\section*{Appendix: Shallow-Depth Variational Quantum Hypothesis Testing}
\section{Diamond norm estimation}
We have two quantum channels $\mathcal{E}_0:L(\mathcal{H})\to L(\mathcal{H})$ and $\mathcal{E}_1:L(\mathcal{H})\to L(\mathcal{H})$ where $L(\mathcal{H})$ is the set of linear operators from Hilbert space $\mathcal{H}$ to $\mathcal{H}$. We denote subset of $L(\mathcal{H})$ that are density operators as $D(\mathcal{H})$. We proceed on the task of channel discrimination, and as highlighted in the main text, we pick a state $\rho\in D(\mathcal{H}\otimes\mathcal{H})$ as a probe to the map $\mathcal{E}_{i}\otimes\mathbb{I}_{\mathcal{H}}$ ($i = 0$ or $1$) and then use a POVM $\{\Gamma,\mathbb{I} - \Gamma\}$ for classification. We present this as an algorithm for estimation of diamond distance as well.\\
Using the parameterized circuit shown in Fig. 1 of the letter, we define the state after applying the map as $\rho_{i,\pmb{\theta}} = (\mathcal{E}_i\otimes \mathbb{I})(\rho_{\pmb{\theta}})$ and $p_i = \text{Tr}(\Gamma\rho_{i,\pmb{\theta}})$. If the unitary $V(\pmb{\phi})$ encodes the Naimark extension of the POVM $\{\Gamma,\mathbb{I} - \Gamma\}$, we obtain $p_i = \text{Tr}\left((\ketbra{0}{0}\otimes I)V(\pmb{\phi})(\ketbra{0}{0}\otimes \rho_{i,\pmb{\theta}})V(\pmb{\phi})^\dagger\right)$. From this we define the cost function in the following equation which is bounded above by the diamond distance since the diamond distance is defined as the supremum of $\|\rho_{0,\pmb{\theta}} - \rho_{1,\pmb{\theta}}\|_1$ over all possible $\rho\in D(\mathcal{H}\otimes\mathcal{H})$.
\begin{equation}\label{eq:TDest}
    T^{\mathrm{ve}}_{\pmb{\phi}}(\pmb{\theta}) = |p_0-p_1| \leq \frac{1}{2}\|\rho_{0,\pmb{\theta}} - \rho_{1,\pmb{\theta}}\|_1 \leq \frac{1}{2}\|\mathcal{E}_0 - \mathcal{E}_1\|_\diamond
\end{equation}
\begin{algorithm}[H]
\SetAlgoLined
Take input as $\mathcal{E}_0$ and $\mathcal{E}_1$ which are CPTP maps from density matrices in $\mathcal{H}$ to density matrices in $\mathcal{H}$.\\
Initialize $\pmb{\theta}$ and $\pmb{\phi}$ which are parameters for $U(\pmb{\theta})$ and $V(\pmb{\phi})$.\\
Define convergence condition for cost function $T^{\mathrm{ve}}_{\pmb{\phi}}(\pmb{\theta})$.\\
\While{$T^{\mathrm{ve}}_{\pmb{\phi}}(\pmb{\theta})$ has not converged}{
    \For{$\pmb{\theta}_i = \pmb{\theta}+\Delta\pmb{\theta}_i$ and $\pmb{\phi}_i = \pmb{\phi}+\Delta\pmb{\phi}_i$ for some set of $\Delta\pmb{\theta}_i$ and $\Delta\pmb{\phi}_i$}{
        Run circuit with parameters $\pmb{\theta}_i$ and $\pmb{\phi}_i$ with CPTP map as $\mathcal{E}_0\otimes\mathbb{I}_n$\\
        Using measurements on $q_0$ obtain value of $p_0$.\\
        Run circuit with parameters $\pmb{\theta}_i$ and $\pmb{\phi}_i$ with CPTP map as $\mathcal{E}_1\otimes\mathbb{I}_n$\\
        Using measurements on $q_0$ obtain value of $p_1$.\\
        Assign $T^{\mathrm{ve}}_{\pmb{\phi}_i}(\pmb{\theta}_i) = |p_0 - p_1|$}
    Using $T^{\mathrm{ve}}_{\pmb{\phi}_i}(\pmb{\theta}_i)$ update $\pmb{\theta}$ and $\pmb{\phi}$ with a classical optimizer.\\
    Check convergence of $T^{\mathrm{ve}}_{\pmb{\phi}}(\pmb{\theta})$.
}
Assign estimated diamond norm $= 2T^{\mathrm{ve}}_{\pmb{\phi}}(\pmb{\theta})$\\
\textbf{return} final values of $\pmb{\theta}$, $\pmb{\phi}$ and estimated diamond norm.
\caption{Variational quantum algorithm to estimate diamond distance}
\end{algorithm}
\begin{figure}[ht]
    \centering
    \includegraphics[width=0.85\textwidth]{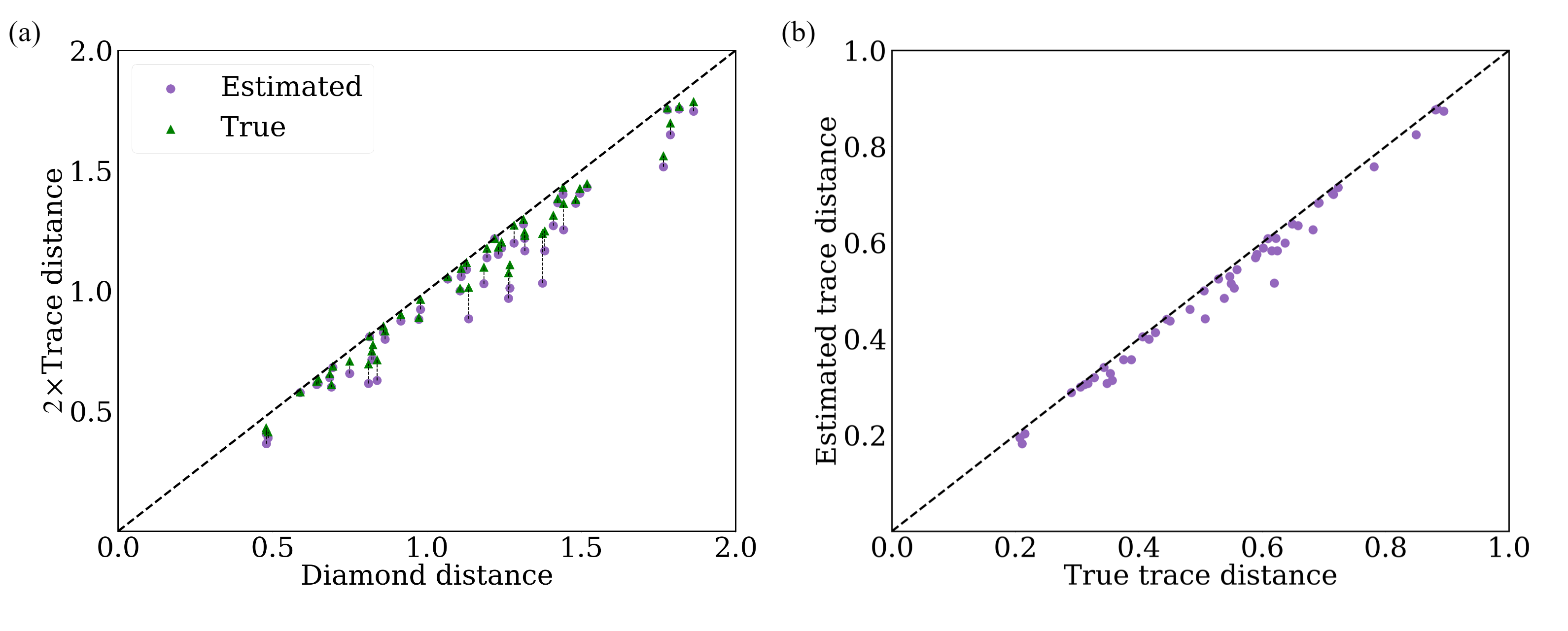}
    \caption{Estimation for diamond distance between the identity map and an arbitrary $2$ qubit unitary map. (a) Shows how well the value of estimated $T^{\mathrm{ve}}_{\pmb{\phi}}(\pmb{\theta})$ and the true trace distance $T(\mathcal{E}_0(\rho_{\pmb{\theta}}),\mathcal{E}_1(\rho_{\pmb{\theta}}))$ for the optimized state matches against the analytical value. (b) Shows how well the estimated trace distance matches against the true trace distance for the state $\rho_{\pmb{\theta}}$}
    \label{fig:diamondNormU}
\end{figure}
\begin{figure}[ht]
    \centering
    \includegraphics[width=0.85\textwidth]{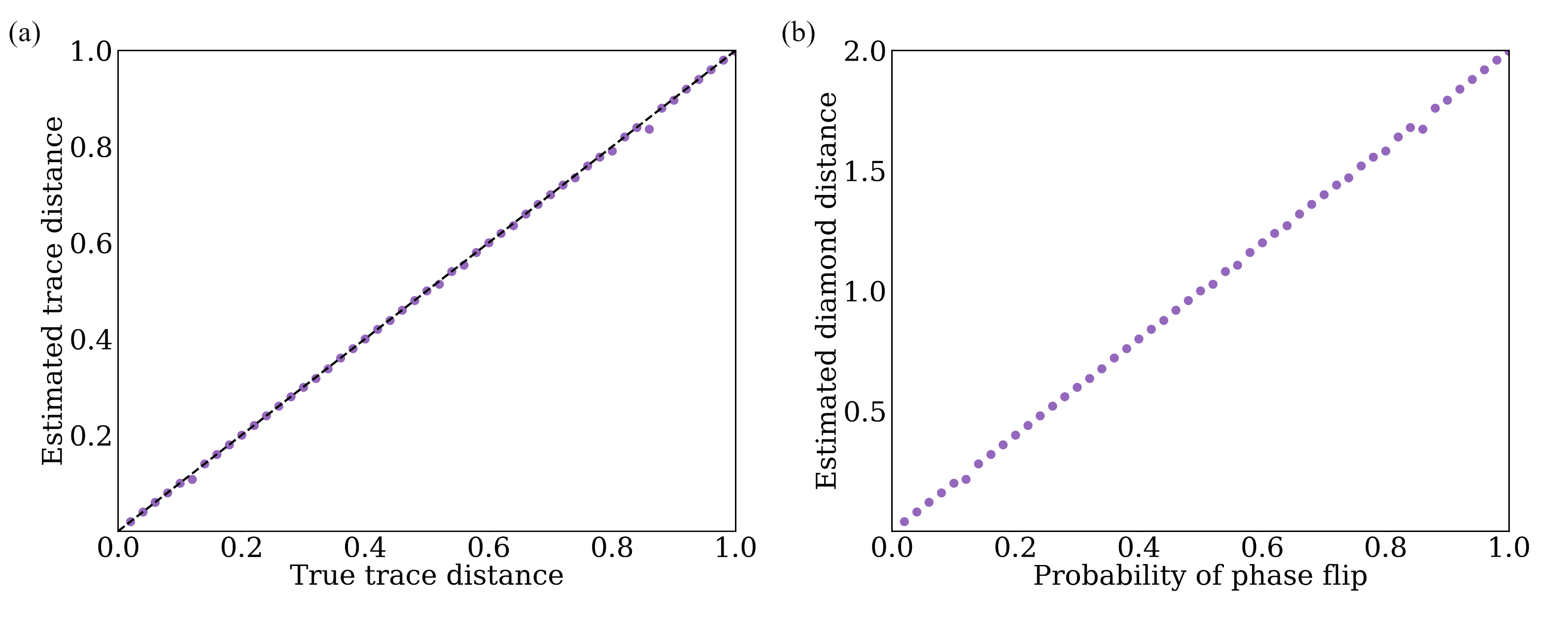}
    \caption{Estimation for diamond distance between the identity map and a phase flip map. (a) Shows how well the estimated trace distance matches against the true trace distance for the state $\rho_{\pmb{\theta}}$. (b) Shows the result of the estimated diamond norm against the probability of phase flip.}
    \label{fig:diamondNormBitflip}
\end{figure}
There are existing ways to find the diamond distance in $\mathcal{O}(poly(\dim{\mathcal{H}}))$ since it is a convex optimization problem \cite{ben2009complexity} but this clearly will grow exponentially in size as we increase the number of qubits. Our algorithm is clearly scalable for such situations and can produce results using a hardware efficient ansatz.\\
\indent Figures \ref{fig:diamondNormU} and \ref{fig:diamondNormBitflip} show results of simulations for estimating of diamond norm. The diamond distance between a unitary map and the identity map is the diameter of the circle which is able to contain all the eigenvalues of the unitary operation \cite{benenti2010computing}. For Kraus maps such as the phase flip map of form $\mathcal{E}(\rho) = (1-p)\rho + pZ\rho Z$ have a diamond distance of $2p$ from the identity map \cite{watrous2007distinguishing}.\\
\indent The simulations were carried out using QuTip \cite{JOHANSSON20121760,JOHANSSON20131234}. The circuit had five qubits with one of them being used in the trace distance estimation subroutine and the other four being used for state preparation after which the quantum map is applied on the first two of the four. The ansatz used was a hardware efficient ansatz \cite{Kandala_2017} where we have single qubit rotations followed by an entangling operation which in this case is made of only CZ gates.

\section{Information content for optimal hypothesis testing}
In this section we analyze the minimum amount of information content required for reaching an $\epsilon$-neighbourhood of the state which is optimal for hypothesis testing. This is to understand how much expressibility the ansatz used for state preparation requires to get a sufficiently good probe state. Here we consider the quantum channels $\mathcal{E}_0:L(\mathcal{H})\to L(\mathcal{H})$ and $\mathcal{E}_1:L(\mathcal{H})\to L(\mathcal{H})$ which satisfy the relation $\|((\mathcal{E}_0-\mathcal{E}_1)\otimes\mathbb{I})\rho_\mathrm{ideal}\|_1 = \|\mathcal{E}_0-\mathcal{E}_1\|_\diamond = d$.\\
To be able to reach the state $\rho_\mathrm{ideal}$ requires providing some classical information which encodes the quantum control problem of approaching this state. We assume that $\rho_{\mathrm{ideal}}\in \mathcal{W}^+$ which is the set of time-polynomial reachable states in $D(\mathcal{H}\otimes\mathcal{H})$ using a certain control scheme $\gamma(t)$ where the set of reachable states are $\mathcal{W}$. Let there be a state $\rho_\mathrm{real}$ in the epsilon neighbourhood of $\rho_\mathrm{ideal}$ 
\begin{equation}
    \|\rho_\mathrm{ideal} - \rho_\mathrm{real}\|_1 \leq \epsilon.
\end{equation}
Using the results of \cite{PhysRevLett.113.010502}, the number of bits which can encode the control pulse $\gamma(t)$ must satisfy
\begin{equation}\label{eq:minInfo}
    b_\gamma \geq \dim(\mathcal{W}^+)\log_2\left(\frac{1}{\epsilon}\right)
\end{equation}
This can be understood by dividing the space of $\mathcal{W}^+$ into epsilon balls which would have a volume scaled as $\varepsilon^{\dim(\mathcal{W}^+)}$ with respect to the total volume of the space. The information content in $\gamma(t)$ must be enough to specify the epsilon ball which we wish to be in which leads directly to \eqref{eq:minInfo}. We define the operator $T = (\mathcal{E}_1 - \mathcal{E}_0)\otimes\mathbb{I}$ and the operator norm $\|T\|_1 = \sup_\rho\|T\rho\|$ for $\rho\in D(\mathcal{H}\otimes \mathcal{H})$. We can take note of the following from the properties of the diamond norm \cite{aharonov1998quantum}.
\begin{gather}
    \|T\rho_\mathrm{real}\|_1 \leq \|T\rho_\mathrm{ideal}\| = d\\
    \|T(\rho_\mathrm{ideal}-\rho_\mathrm{real})\|_1 \leq \|T\|_1\|(\rho_\mathrm{ideal}-\rho_\mathrm{real})\|_1 \leq d\varepsilon
\end{gather}
Using the triangle inequality for the 1-norm we obtain
\begin{equation}
     \|T(\rho_\mathrm{ideal}-\rho_\mathrm{real})\|_1 \geq \left|\|T\rho_\mathrm{ideal}\|_1 - \|T\rho_\mathrm{real}\|_1\right| = d - \|T\rho_\mathrm{real}\|_1 \geq 0
\end{equation}
Combining the above inequalities, we obtain
\begin{equation}
    0 \leq d - \|T\rho_\mathrm{real}\|_1 \leq d\epsilon \implies \|T\rho_\mathrm{real}\|_1 \in [d(1-\varepsilon),d]
\end{equation}
Notably the minimal $b_\gamma$ is independent of $d$. This expression shows that being in the $\varepsilon$ neighbourhood of the state which maximizes trace distance implies that the trace distance now lies between $d(1-\varepsilon)$ and $d$ irrespective of the value of $d$ and will require the same amount of classical information. The computational toughness will arise in the fact that if $d$ is small enough, even the best possible state is unable to tell apart the two channels.

\section{Simultaneous optimization in the algorithm}
In this section we will prove that the method of simultaneous optimization employed for the algorithm used for variational quantum hypothesis testing is valid. We first begin with defining the form of our estimated trace distance. There are ways of estimating trace distance using a variational quantum algorithm as shown in \cite{chen2021variational} and \cite{https://doi.org/10.48550/arxiv.2108.08406}. The main clue lies in the following definition of trace distance.
\begin{equation}\label{eq:varTr}
T(\rho_0,\rho_1) = \sup_{P\leq\mathbb{I}}(\text{Tr}(P(\rho_0-\rho_1)))
\end{equation}
We can variationally optimize the POVM $P$ to obtain an estimate of the trace distance. To do this using a unitary operation, we must embed the POVM into the unitary operator. For this the Naimark extension can be used \cite{wilde2013sequential}.
\begin{theorem}[Naimark extension]
For any POVM $\{\Gamma_i\}_{i\in O}$ acting on a system $S$, there exists a unitary $U_{PS}$ (acting on a probe system $P$ and the system $S$) and an orthonormal basis $\{\ket{i}_P\}_{i\in O}$ such that
\begin{equation}
    \text{Tr}\left((\ketbra{i}{i}\otimes I_S)U_{PS}(\ketbra{i}{i}\otimes \rho_S)U_{PS}^\dagger\right) = \text{Tr}(\Gamma_i\rho_S)
\end{equation}
\end{theorem}
As pointed out in \cite{wilde2013sequential}, the two-outcome POVM $\{\Gamma,\mathbb{I}-\Gamma\}$ can be encoded in the following unitary with the probe system being a qubit.
\begin{equation}\label{eq:naimark2}
    U_{PS} = \mathbb{I}_P\otimes(\sqrt{\Gamma})_S + i(\sigma_Y)_P\otimes(\sqrt{\mathbb{I} - \Gamma})_S
\end{equation}
Let us define a parameterized unitary $V(\phi)$ which acts over both the probe and the system. We define the following quantity as an estimate of trace distance,
\begin{gather}
    T_{\pmb{\phi}}(\rho_0,\rho_1) = |p_0 - p_1|\\
    p_i = \text{Tr}\left((\ketbra{0}{0}\otimes I)V(\pmb{\phi})(\ketbra{0}{0}\otimes \rho_i)V(\pmb{\phi})^\dagger\right)\quad i\in\{0,1\}
\end{gather}
On combining theorem 4.1 and equation \eqref{eq:varTr}, we get that $\forall{\pmb{\phi}}(T_{\pmb{\phi}}(\rho_0,\rho_1) \leq T(\rho_0,\rho_1))$. This is the main crux of using a variational algorithm for estimating trace distance in \cite{chen2021variational}. As an extension to this, we define the following cost function for our algorithm where we use an additional qubit as the probe subsystem.
\begin{gather}
    T^{\mathrm{ve}}_{\pmb{\phi}}(\pmb{\theta}) = |p_0(\pmb{\theta},\pmb{\phi}) - p_1(\pmb{\theta},\pmb{\phi})|\\
    p_i(\pmb{\theta},\pmb{\phi}) = \text{Tr}\left((\ketbra{0}{0}\otimes I)V(\pmb{\phi})(\ketbra{0}{0}\otimes \mathcal{E}_i(\rho(\pmb\theta)))V(\pmb{\phi})^\dagger\right)\quad i\in\{0,1\}
\end{gather}
Our optimization procedure will have to optimize both $\pmb{\theta}$ and $\pmb{\phi}$ for obtaining the best possible state preparation and measurement.
Let us define $\pmb{\theta}_0$ and $\pmb{\phi}_0$ as follows
\begin{align}
    \pmb{\theta}_0 &= \argmax_{\pmb{\theta}}(\|\mathcal{E}_1(\rho_{\pmb{\theta}}) - \mathcal{E}_0(\rho_{\pmb{\theta}})\|_1)\\
    \pmb{\phi}_0 &= \argmax_{\pmb{\phi}}(T^{\mathrm{ve}}_{\pmb{\phi}}(\pmb{\theta}_0))
\end{align}
Here $\pmb{\theta}_0$ optimizes toward the state that saturates the Holevo-Helstrom bound \cite{holevo1973bounds,Helstrom1969}. As per the definition of the optimization problem of trace distance estimation \cite{chen2021variational}, $\pmb{\phi}_0$ represents the best parameters to estimate the trace distance for $\rho_{\pmb{\theta}_0}$.
Now let us define the parameters obtained by a complete optimization as follows
\begin{equation}
    \tilde{\pmb{\theta}}, \tilde{\pmb{\phi}} = \argmax_{\pmb{\theta},\pmb{\phi}}(\text{Tr}_{\pmb{\phi}}(\mathcal{E}_1(\rho_{\pmb{\theta}})) - \text{Tr}_{\pmb{\phi}}(\mathcal{E}_0(\rho_{\pmb{\theta}})))
\end{equation}
Clearly $T^{\mathrm{ve}}_{\pmb{\phi}}(\pmb{\theta}) \leq \|\mathcal{E}_1(\rho_{\pmb{\theta}}) - \mathcal{E}_0(\rho_{\pmb{\theta}})\|_1/2$. Our task now would be to verify if $\tilde{\pmb{\theta}}, \tilde{\pmb{\phi}} \equiv \pmb{\theta}_0, \pmb{\phi}_0$, to see whether the global optimization reaches a meaningful result.
\begin{claim}
Under the assumption that for all $\pmb{\theta}$, there exists a $\pmb{\phi}$ which satisfies \[T^{\mathrm{ve}}_{\pmb{\phi}}(\pmb{\theta}) = T(\rho_{1,\pmb{\theta}},\rho_{0,\pmb{\theta}}),\]
we can claim the equivalence $\tilde{\pmb{\theta}}, \tilde{\pmb{\phi}} \equiv \pmb{\theta}_0, \pmb{\phi}_0$.
\end{claim} 
\begin{proof}
From the assumption we have taken, it is quite clear that $T^{\mathrm{ve}}_{\pmb{\phi}_0}(\pmb{\theta}_0) = \|\mathcal{E}_1(\rho_{\pmb{\theta}_0}) - \mathcal{E}_0(\rho_{\pmb{\theta}_0})\|_1$. Along with this, since the parameters $\tilde{\theta},\tilde{\phi}$ are from an optimization of $T^{\mathrm{ve}}_{\pmb{\phi}}(\pmb{\theta})$, we must have the following inequality hold
\[T^{\mathrm{ve}}_{\tilde{\pmb{\phi}}}(\tilde{\pmb{\theta}}) \geq T^{\mathrm{ve}}_{\pmb{\phi}_0}(\pmb{\theta}_0)\]
Now from the assumption that we have taken we can make the following claim
\[T^{\mathrm{ve}}_{\tilde{\pmb{\phi}}}(\tilde{\pmb{\theta}}) = \|\mathcal{E}_1(\rho_{\tilde{\pmb{\theta}}}) - \mathcal{E}_0(\rho_{\tilde{\pmb{\theta}}})\|_1\]
If this doesn't hold, there will exist some $\pmb{\phi}$ which gives the exact trace distance and the $TD$ function always is less than the trace bound hence resulting in a contradiction. Hence we have the following hold
\[\|\rho_{1,\tilde{\pmb{\theta}}} - \rho_{0,\tilde{\pmb{\theta}}}\|_1 \geq \|\rho_{1,\pmb{\theta}_0} - \rho_{0,\pmb{\theta}_0}\|_1\]
Let us assume that $\|\rho_{1,\tilde{\pmb{\theta}}} - \rho_{0,\tilde{\pmb{\theta}}}\|_1 > \rho_{1,\pmb{\theta}_0} - \rho_{0,\pmb{\theta}_0}\|_1$. This would contradict the fact that $\pmb{\theta}_0$ is a parameter that saturates the trace distance. Hence we finally get the following hold.
\[\|\rho_{1,\tilde{\pmb{\theta}}} - \rho_{0,\tilde{\pmb{\theta}}}\|_1 = \|\rho_{1,\pmb{\theta}_0} - \rho_{0,\pmb{\theta}_0}\|_1\]
Hence both these parameters saturate the Holevo bound and also they have the perfect trace distance estimators, hence proving their equivalence.
\end{proof}
While the assumption in the above claim requires $V(\pmb{\phi})$ to be able to reach the optimal POVM's Naimark extension for all $\rho_{\pmb{\theta}}$, this does show that the optimization procedure is sound and produces meaningful results. In essence, we will have to optimize our estimated trace distance since the real trace distance is not as easy to calculate but this optimization will end up optimizing the true trace distance as well as the estimate of trace distance toward the true value. This has been reflected in the results we show for variational quantum illumination using Gaussian states.
\begin{figure}
    \centering
    \includegraphics[width=0.95\textwidth]{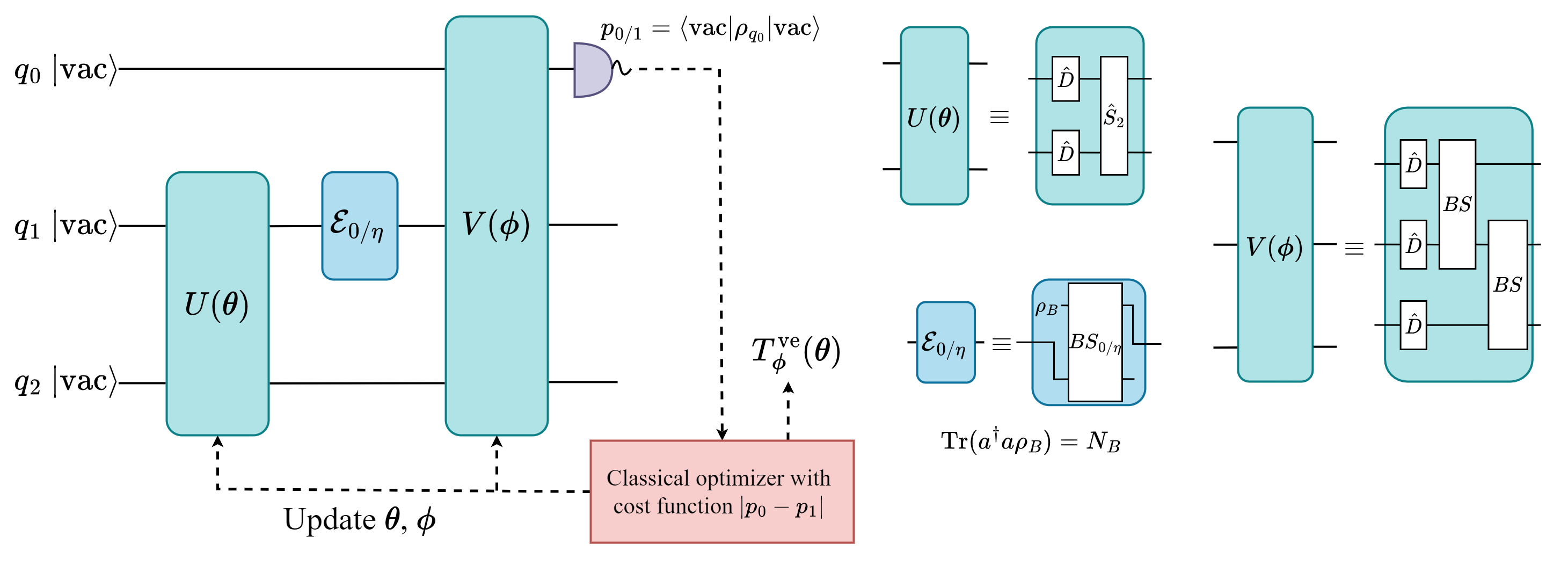}
    \caption{Description of the quantum circuit for variational quantum illumination. The state is prepared over the signal and idler mode using an ansatz which consists of displacements and two-mode squeezing. The ansatz shown for the measurement section consists of displacements and beam splitters. The simulations contain controlled phase gates for in the measurement section as well. The readout can be either measure the probability it is vacuum or the expectation value of the parity operator.}
    \label{fig:VQIcircuit}
\end{figure}
\section{Naimark dilation for Gaussian states}
When we are working with purely Gaussian states, we have the limitation that the unitary $V$ must also be a Gaussian unitary. This puts a fundamental limitation on the form it can take, given that no two Gaussian states are orthogonal \cite{ferraro2005gaussian}. The overlap can be made as small as needed, but true orthogonality is impossible and hence we cannot reach the true canonical Naimark extension \cite{paris2012modern} if we choose to use only Gaussian operations. The isometry of the canonical Naimark extension in this case is given as follows which is not Gaussian.
\begin{equation}
    V_\mathrm{canonical} = \mathbb{I}_{q_0}\otimes(\sqrt{\Gamma})_{SI} + i(-\ketbra{0}{1}_{q_0} + \ketbra{1}{0}_{q_0} + \mathbb{I} - \ketbra{0}{0} - \ketbra{1}{1})\otimes(\sqrt{\mathbb{I} - \Gamma})_{SI}
\end{equation}
On the other hand, we can frame this as trying to perform a measurement on some $n$ mode state by entangling it to a 1 mode system. This can be written as follows.
\begin{equation}
    \text{Tr}(V^\dagger(P_0\otimes\mathbb{I}_n)V(\ketbra{0}{0}\otimes\rho)) = \text{Tr}(P\rho)
\end{equation}
Here $P_0$ is a a projection which is in $L(\mathcal{H})$ where $\mathcal{H}$ is the Hilbert space of a single mode of light. $P$ is a projection in $L(\mathcal{H}^{\otimes n})$ and $\rho$ is a density operator in $D(\mathcal{H}^{\otimes n})$ and $\mathbb{I}_n$ is the identity operator in $L(\mathcal{H}^{\otimes n})$. Since we want the above equation to hold for all $\rho$, we can rewrite it as follows.
\begin{equation}
    V^\dagger (P_0\otimes\mathbb{I}_n) V = \mathbb{I}_1\otimes P
\end{equation}
Here we are applying a transformation from one linear operator to another which means that as long as the norm of both $P_0\otimes\mathbb{I}_n$ and $\mathbb{I}_1\otimes P$ are equal, we can find a $V$. We can construct a $V = (\mathrm{SWAP}_{1,2}\otimes\mathbb{I}_{n-1})(\mathbb{I}_1\otimes V')$ where it performs a swap between the first two modes and then does some unitary $V'$ only on the subsystem of $n$ modes. This transforms the measurement from the space of the first mode to the space of the $n$ modes.
\begin{equation}
    (\mathbb{I}_1\otimes V')^\dagger(\mathrm{SWAP}_{1,2}\otimes\mathbb{I}_{n-1})^\dagger(P_0\otimes\mathbb{I}_n)(\mathrm{SWAP}_{1,2}\otimes\mathbb{I}_{n-1})(\mathbb{I}_1\otimes V') = \mathbb{I}_1\otimes V'^\dagger(P_0\otimes\mathbb{I}_{n-1})V'^\dagger
\end{equation}
This shows that we can construct any projection of form $P = V'^\dagger(P_0\otimes\mathbb{I}_{n-1})V'^\dagger$ where $V'$ is Gaussian since the swap operation between two modes can be trivially represented as a passive Gaussian operation. This recipe shows us that while we may not be able to construct the canonical Naimark extension of the optimal POVM, we can construct a Naimark extension that performs a POVM on the $n$ mode subspace using a single-mode ancillary measurement.
\section{Multiple optima in Gaussian quantum illumination}
As can be seen in Fig. 2 of the main text, there are certain states which do not have 100\% fidelity with the TMSV yet happen to have equal performance in the QFI, chernoff bound and the trace distance. These states are largely accessible due to the constraint only being placed on the value of signal photon number $N_S$.\\
We pick the example of the state obtained in the case of $N_B = 1$. We perform a Schmidt decomposition on this state and compare it with the TMSV. Equal Schmidt values imply that one state can be obtained from the other using only local transformations implying equal entanglement as well. It turns out that this is not the case and the TMSV is different from the obtained optimal state despite both being equally good for the hypothesis testing task. This implies that there are clearly multiple possible Gaussian states which are optimal for the hypothesis testing task.

The five largest Schmidt values for TMSV with $N_S = 0.1$ are $6.20921323\times 10^{-5}, 6.83013455\times 10^{-4}, 7.51314801\times 10^{-3}, 8.26446281\times 10^{-2}, 9.09090909\times 10^{-1}$. The TMSV has a von-Neumann entropy of $0.33509970612111517$.\\
The five largest Schmidt values for the obtained optimal state at $N_B = 1$ with $N_S = 0.1$ are $6.03874183\times 10^{-5}, 6.69019663\times 10^{-4}, 7.41192949\times 10^{-3}, 8.21152229\times 10^{-2}, 9.09737450\times 10^{-1}$. This state has a von-Neumann entropy of $0.3332223308457541$.\\
\section{Testing unknown values in variational quantum illumination}
\begin{figure}[ht]
    \centering
    \includegraphics[width=0.9\textwidth]{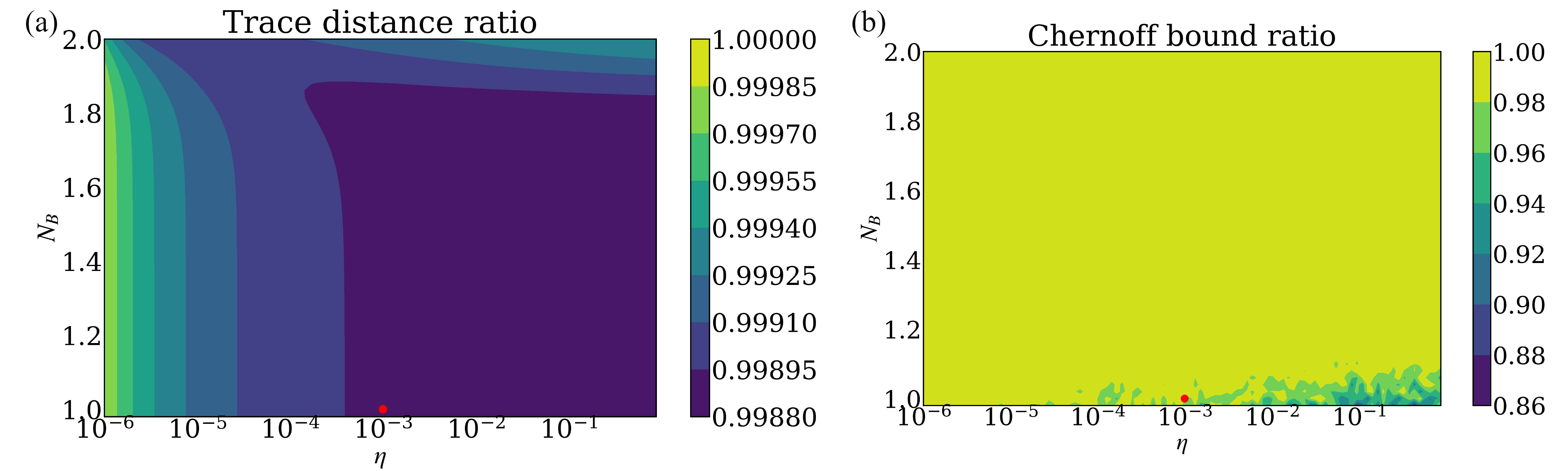}
    \caption{Comparision of performance of a specific optimized state over a range of $N_B$ and $\eta$. The state we choose is obtained by optimizing with $N_B=1$ and $\eta=10^{-3}$ (indicated by a red dot in the image) which as shown previously is quite different from the TMSV. (a) Shows the ratio of trace distance when using the chosen state to that of the trace distance with TMSV (b) Shows the ratio of Chernoff bound when using the chosen state to that of the trace Chernoff bound with TMSV. Both of these are measures of performance for a good hypothesis test in the symmetric case giving an idea that even over unknown values of $N_B$ and $\eta$ we obtain states which do perform optimally.}
    \label{fig:performanceratio}
\end{figure}
\indent The optimization protocol we use happens to find the optimal state for discriminating two given maps $\mathcal{E}_0$ and $\mathcal{E}_1$. However in the case of quantum illumination, where these maps are dependent on some continuous parameters (background radiation $N_B$ and reflectivity $\eta$), one might want to check the applicability of the optimized state for varying values of $N_B$ and $\eta$. This is checked by seeing how well a fixed state performs compared to the known optimal TMSV state as can be seen in Fig. \ref{fig:performanceratio}.

\section{Simulation code and data}
All the simulation code and data for variational Gaussian quantum illumination, general QHT, and multiple hypothesis testing can be found here: \url{https://github.com/mahadevans2432/Variational-QHT}.

\end{document}